\documentclass[10pt,conference]{IEEEtran}
\usepackage{color}
\usepackage[svgnames]{xcolor}
\usepackage{graphicx}
\usepackage{amsmath}
\usepackage{amsthm}
\usepackage{mathtools}
\usepackage{makecell}
\usepackage{physics}
\usepackage{bm}
\usepackage{hyperref}
\usepackage{multirow}

\usepackage{physics}

\usepackage{amsmath}
\usepackage{amssymb}
\usepackage{amsthm}
\usepackage{graphicx}
\usepackage{float}
\usepackage{listings}
\usepackage{xcolor}
\usepackage{subcaption}
\usepackage{tikz}
\usepackage{pgfplots}
\pgfplotsset{compat=1.17}
\usetikzlibrary{arrows.meta, calc, positioning}
\usepackage[ruled, linesnumbered]{algorithm2e}
\usepackage{algpseudocode}
\usepackage{algcompatible}
\usepackage{quantikz}

\usepackage{booktabs}
\usepackage{siunitx}

\usepackage{tikz} 
\usetikzlibrary{matrix,positioning}
\usepackage{amsfonts}
\usepackage{amssymb}
 
\newtheorem{definition}{Definition}[section]

\newtheorem{theorem}{Theorem}[section]

\newtheorem{proposition}{Proposition}

\hypersetup{
   colorlinks = true,
    urlcolor=blue
}
\newtheorem{remark*}{Remark}

\begin{document}

\markboth{First author}
{Explicit Block Encoding of Difference-of-Gaussian Operators on a Periodic Grid}

\title{Explicit Block Encoding of Difference-of-Gaussian Operators on a Periodic Grid}

 \author{\IEEEauthorblockN{First author\IEEEauthorrefmark{1}, Second author\IEEEauthorrefmark{1}, Third author\IEEEauthorrefmark{1}, Fourth author\IEEEauthorrefmark{1}, Last author\IEEEauthorrefmark{1}\IEEEauthorrefmark{3}
 }
 \IEEEauthorblockA{\IEEEauthorrefmark{1} Redacted}}

 \author{\IEEEauthorblockN{Jishnu Mahmud\IEEEauthorrefmark{1}\IEEEauthorrefmark{3}, John Winship\IEEEauthorrefmark{2}, Tom Lash\IEEEauthorrefmark{1}\IEEEauthorrefmark{2}, James Ostrowski\IEEEauthorrefmark{1}, Rebekah Herrman\IEEEauthorrefmark{1} \IEEEauthorrefmark{4}
 }
 \IEEEauthorblockA{\IEEEauthorrefmark{1}Industrial and Systems Engineering, University of Tennessee at Knoxville, Knoxville, TN, USA}
 \IEEEauthorblockN{\IEEEauthorrefmark{2}Vibrint, Hanover, MD, USA}
 
 Email: \IEEEauthorrefmark{3}jmahmud@vols.utk.edu, \IEEEauthorrefmark{4}rherrma2@utk.edu}

\maketitle
\begin{abstract}

The Difference-of-Gaussian (DoG) is a widely used operator across applications, including image processing (feature and edge detection), quantum machine learning, and finite-difference methods (approximations of the Laplacian-of-Gaussian). In this paper, we construct an explicit quantum block encoding of the DoG operator on a periodic grid, exploiting its natural probabilistic structure. The central observation is that the DoG admits a natural decomposition to two normalized Gaussian distributions, each preparable by explicit and efficient circuits, with the negation encoded using a single Pauli-$Z$ gate on a branch-indicator qubit. This enables the operator's block encoding to be directly mapped to the Linear Combination of Unitaries framework without requiring signed amplitude loading, quantum random-access memory, or any other black-box oracles. The proposed method achieves a constant subnormalization factor $\lambda = 2$ independent of the grid size $N$, the spatial dimension $D$, and the stencil width. Additionally, we show that the DoG operator is diagonalized by the discrete Fourier basis, which allows us to derive an exact closed-form expression for the block-encoding success probability in terms of the input signal's power spectrum, weighted by the operator's transfer function. Finally, we prove that the expression reduces to $\mathcal{O}(h^4)$ scaling with respect to grid spacing $h$ as the periodic grid becomes finer. This implementation provides an explicit construction method for a tunable, wide-stencil bandpass filter whose frequency response is controlled by two Gaussian scale parameters.

\end{abstract}

\section{Introduction}\label{sec:introduction}

The Difference-of-Gaussian (DoG) operator is a fundamental bandpass filter widely used in image processing for feature and edge detection \cite{kong2013generalized, polakowski1997computer, mu2016multiscale, assirati2014performing}, and serves as a smooth, discrete approximation to the Laplacian-of-Gaussian (LoG) \cite{marr1980theory}. More broadly, differences of discrete operators are a foundation for finite-difference approximations to partial differential equations (PDEs) \cite{leveque1998finite, liszka1980finite, smith1985numerical, vuik2023numerical, mirzaee2022solution} and appear throughout computational signal analysis \cite{sotak1989laplacian, neycenssac1993contrast}. In all of these settings, the spatial complexity of representing the underlying system grows as $\mathcal{O}(N^D)$, where $N$ is the number of grid points per dimension and $D$ is the spatial dimension \cite{bungartz2004sparse, sogabe2022krylov}, rendering high-resolution, multi-dimensional problems intractable on classical hardware.

Quantum computing offers a potential path for overcoming these dimensional bottlenecks. Several quantum edge-detection and filtering algorithms have been proposed based on LoG, DoG, discrete gradient, and shift-based operators \cite{wu2022quantum, sundani2019identification, li2020quantum, yao2017quantum, chetia2021quantum, liu2023quantum, zhou2019quantum, yuan2019fast, yuan2024quantum}. Many of these approaches employ basis-state image encodings, such as the Novel Enhanced Quantum Representation (NEQR) \cite{fan2023quantum, wu2022quantum, li2020quantum, chetia2021quantum, liu2023quantum, zhou2019quantum}, in which pixel intensities are stored in computational-basis registers. While versatile, NEQR-based methods implement filtering via quantum arithmetic circuits (such as adders, comparators, and multi-controlled rotations), whose circuit depth and gate count scale with both the image size and the bit depth of the pixel values. Critically, these constructions treat the filter coefficients as generic numerical data, without exploiting any special structure the coefficients may possess.

An alternative method involves encoding continuous data as the probability amplitudes of a quantum state, a technique known as amplitude encoding, allowing linear filtering operations via block encoding \cite{gilyen2019quantum} to be performed on the state. In this framework, a non-unitary operator is embedded as the top-left block of a larger unitary, which enables the resultant operator to be integrated with modern quantum linear-algebra primitives such as the quantum singular value transform (QSVT) and qubitization \cite{gilyen2019quantum, berry2007efficient, childs2017quantum}. However, standard block-encoding constructions typically rely on black-box oracle access such as quantum random access memory (QRAM) or unary iteration to load the matrix entries of~$A$ \cite{berry2019qubitization, babbush2018encoding}. Recent work has demonstrated that such general-purpose, structure-agnostic methods incur intractable gate costs even for modest system sizes, and that the T-gate complexity of unstructured oracles can be large \cite{kuklinski2025efficient, zecchi2026block}.

In light of these observations, research has moved beyond the generalized block-encoding formalism to develop explicit circuits that implement encoding methods exploiting the mathematical structure of operators. For the discrete Laplacian operator, Camps et al. proposed a quantum circuit for banded circulant encodings \cite{camps2024explicit}, and Kharazi et al. designed methods for Dirichlet and Neumann boundary conditions \cite{kharazi2025explicit}. Subsequently, Sturm and Schillo \cite{sturm2025efficient} proposed a block-encoding circuit that achieves an optimal subnormalization factor for the standard 3-point Laplacian stencil using only Clifford gates. A common focus in this line of work is that efficient and explicit block encodings require the circuit to respect the mathematical structure of the target operator \cite{kuklinski2025efficient}.

In this work, we develop an efficient block-encoding for the DoG operator that achieves a constant subnormalization factor independent of the grid size or dimension. The central observation of this paper is that the DoG operator possesses a natural probabilistic structure that can be directly leveraged for efficient quantum implementation. Specifically, the discrete DoG stencil has the form $A_h = \sum_{t\in T}(p_t - q_t)\, S_t$, where $\{p_t\}$ and $\{q_t\}$ are two normalized discrete Gaussian probability distributions on a discrete grid and $T$ and $S_t$ are cyclic shift operators. We normalize $p$ and $q$ post-discretization, asserting the condition $\sum_t p_t = \sum_t q_t = 1$ and nonnegativity, which consequently causes the 1-norm of the coefficient vector to satisfy $\sum_t |p_t - q_t| \leq 2$. The operator, therefore, decomposes as a difference of two separately preparable components. We observe that the structure maps directly onto the Linear Combination of Unitaries (LCU) framework \cite{berry2015simulating, childs2012hamiltonian}. The magnitudes $\sqrt{p_t}$ and $\sqrt{q_t}$ are loaded by resource-efficient Gaussian state-preparation circuits \cite{plesch2011quantum, kuklinski2025simpler}, while the sign difference is encoded via a Pauli-$Z$ gate acting on a single indicator qubit. No signed amplitude loading, QRAM, or arithmetic circuitry is required in the proposed method.

Building on this observation, we make the following contributions:

\begin{enumerate}
    \item \textbf{Explicit block encoding.} We construct a quantum circuit that block-encodes $A_h$ with a constant subnormalization factor of $\lambda = 2$, independent of the grid size $N$, dimension $D$, or stencil width $|T|$. The non-Clifford components are the Gaussian state-preparation oracles and controlled shifts, for which efficient, explicit circuits exist \cite{plesch2011quantum, kuklinski2025simpler, draper2004logarithmic}.
    \item \textbf{Spectral characterization.} We show that $A_h$ is diagonalized by the discrete Fourier basis with eigenvalues given by the difference of the Fourier transforms of~$ p$ and~$ q$, providing an exact description of its bandpass behavior and confirming Hermiticity under symmetric kernels.
    \item \textbf{Success probability analysis.} We derive an exact closed-form expression for the success probability in terms of the input signal's power spectrum weighted by the operator's transfer function. For smooth input functions, this expression scales as $\mathcal{O}(h^4)$, where $h = 1/N$ is the grid spacing. This matches the theoretically optimal scaling of rigid finite-difference Laplacian encodings \cite{sturm2025efficient}, while offering the additional flexibility of a tunable, wide-stencil bandpass filter whose frequency response is controlled by the two Gaussian scale parameters.
\end{enumerate}

\section{Mathematical Preliminaries}\label{sec:prelims}

We start by defining the mapping of continuous spatial functions onto a discrete Hilbert space and characterizing the DoG operator.  This discretization is a preprocessing step for various applications that aim to apply the DoG filter in quantum systems and, therefore, requires formatting the data for unitary evolution.

\begin{definition}[Grid and Quantum State]\label{def: grid}
Let $\Omega_D = [0, 1)^D \subset \mathbb{R}^D$ be a periodic spatial domain. We discretize this domain into a $D$-dimensional periodic grid $\mathbb{Z}_N^D$ with $N=2^n$ ($n \in \mathbb{N}$) points per dimension.  Here $n$ is the number of qubits assigned to each spatial coordinate, and the discretization scale $h=1/N$. 

The sampling of  a continuous square-integrable function $v: \Omega_D \to \mathbb{C}$ on this grid is $v_h(j) = v(hj)$ for $j \in \mathbb{Z}_N^D$ and the corresponding  quantum state $\ket{v_h}$ is defined as the following superposition state:

\begin{equation}
    \ket{v_h} = \frac{1}{\|v_h\|_2} \sum_{j \in \mathbb{Z}_N^D} v_h(j) \ket{j}
\end{equation}
where $\ket{j} = \ket{j_1} \otimes \dots \otimes \ket{j_D}$ is the computational basis state representing the grid indices.

\end{definition}

In the context of signal processing and partial differential equations, the Difference-of-Gaussian operator serves as a band-pass filter, bringing out features within a specific range of spatial frequencies. Figure \ref{fig: stencil} illustrates discrete periodic grids and highlights the nearest neighbors of the center points in (a) $D=1$ and (b) $D=2$ dimensions.

\begin{figure*}[t]
\centering
\begin{tikzpicture}[
    every node/.style={font=\small},
    grid point/.style={circle, draw, fill=blue!60, inner sep=0pt, minimum size=6pt},
    center point/.style={circle, draw=red!80, fill=red!60, inner sep=0pt, minimum size=8pt},
    stencil point/.style={circle, draw=orange!80, fill=orange!50, inner sep=0pt, minimum size=7pt},
    shift arrow/.style={-{Stealth[length=5pt]}, thick, blue!70!black},
    wrap arrow/.style={-{Stealth[length=4pt]}, thick, red!60!black, dashed},
]

\begin{scope}[shift={(0,-1.5)}]

    \node[font=\normalsize, anchor=south] at (0, 4.5) {(a) 1D periodic grid, $N = 8 = 2^3$};

    \def\R{2.2}
    \def\N{8}

    \draw[gray!30, thick] (0,0) circle (\R);

    \foreach \j in {0,...,7} {
        \pgfmathsetmacro{\angle}{90 - \j * 360 / \N}
        \coordinate (p\j) at (\angle:\R);
    }

    \foreach \j in {2,3,4,5,6} {
        \pgfmathsetmacro{\ang}{90 - \j * 360 / \N}
        \node[grid point] at (p\j) {};
        \node[font=\footnotesize] at ($(\ang:\R) + (\ang:0.4)$) {$\j$};
    }

    \foreach \j in {1, 7} {
        \pgfmathsetmacro{\ang}{90 - \j * 360 / \N}
        \node[stencil point] at (p\j) {};
        \node[font=\footnotesize] at ($(\ang:\R) + (\ang:0.4)$) {$\j$};
    }

    \node[center point] at (p0) {};
    \node[font=\footnotesize] at ($(90:\R) + (90:0.4)$) {$0$};

    \pgfmathsetmacro{\angStart}{90 - 2 * 360/\N + 4}
    \pgfmathsetmacro{\angEnd}{90 - 3 * 360/\N - 4}
    \draw[shift arrow]
        (\angStart:\R+0.15) arc[start angle=\angStart, end angle=\angEnd, radius=\R+0.15];
    \node[font=\footnotesize, blue!70!black] at
        ({90 - 2.5*360/\N}:\R+0.55) {$S_{+1}$};

    \pgfmathsetmacro{\angWS}{90 - 7 * 360/\N + 4}
    \pgfmathsetmacro{\angWE}{90 - 0 * 360/\N - 4}
    \draw[shift arrow, densely dashed]
        (\angWS:\R+0.15) arc[start angle=\angWS, end angle=\angWE, radius=\R+0.15];

\end{scope}

\begin{scope}[shift={(4.2, -0.5)}]
    \node[center point, label={[font=\scriptsize]right: center $j$}] at (0, 0.4) {};
    \node[stencil point, label={[font=\scriptsize]right: stencil ($t \in T$)}] at (0, 0) {};
    \node[grid point, label={[font=\scriptsize]right: grid point}] at (0, -0.4) {};
\end{scope}

\begin{scope}[shift={(7.5, 0.6)}]

    \node[font=\normalsize, anchor=south] at (2.25, 2.4) {(b) 2D periodic grid, $N = 4 = 2^2$};

    \def\sp{1.5}

    \foreach \i in {0,...,3} {
        \foreach \j in {0,...,3} {
            \coordinate (g\i\j) at (\j*\sp, -\i*\sp);
        }
    }

    \foreach \i in {0,...,3} {
        \draw[gray!30] (0, -\i*\sp) -- (3*\sp, -\i*\sp);
        \draw[gray!30] (\i*\sp, 0) -- (\i*\sp, -3*\sp);
    }

    \foreach \i in {0,...,3} {
        \foreach \j in {0,...,3} {
            \node[grid point] at (g\i\j) {};
        }
    }

    \node[center point] at (g11) {};

    \node[stencil point] at (g01) {};
    \node[stencil point] at (g21) {};
    \node[stencil point] at (g10) {};
    \node[stencil point] at (g12) {};

    \foreach \j in {0,...,3} {
        \node[font=\footnotesize, above=2pt] at (\j*\sp, 0) {$\j$};
    }
    \foreach \i in {0,...,3} {
        \node[font=\footnotesize, left=2pt] at (0, -\i*\sp) {$\i$};
    }

    \draw[wrap arrow]
        ($(g23) + (0.15, 0)$)
        .. controls ($(g23) + (0.7, 0.3)$) and ($(g23) + (0.7, -1.0)$) ..
        ($(g23) + (0.4, -0.8)$)
        -- ($(g20) + (0.4, -0.8)$)
        .. controls ($(g20) + (-0.7, -0.8)$) and ($(g20) + (-0.7, 0.2)$) ..
        ($(g20) + (-0.15, 0)$);

    \draw[wrap arrow]
        ($(g31) + (0, -0.15)$)
        .. controls ($(g31) + (-0.3, -0.7)$) and ($(g31) + (-1.2, -0.7)$) ..
        ($(g31) + (-1.0, -0.4)$)
        -- ($(g01) + (-1.0, -0.4)$)
        .. controls ($(g01) + (-1.0, 0.5)$) and ($(g01) + (-0.2, 0.5)$) ..
        ($(g01) + (0, 0.15)$);

    \draw[{Stealth[length=4pt]}-{Stealth[length=4pt]}, thick, gray!60]
        (-0.4, 0.55) -- (3*\sp+0.4, 0.55);
    \node[font=\footnotesize, gray!60, above] at (1.5*\sp, 0.55) {dimension 1};

    \draw[{Stealth[length=4pt]}-{Stealth[length=4pt]}, thick, gray!60]
        (3*\sp+0.6, 0.2) -- (3*\sp+0.6, -3*\sp-0.2);
    \node[font=\footnotesize, gray!60, rotate=-90, above] at (3*\sp+0.6, -1.5*\sp) {dimension 2};

\end{scope}

\end{tikzpicture}
\caption{Periodic grids with dyadic (power of 2) point counts. (a) 1-dimensional grid on $N=8=2^3$ points arranged as a ring to showcase periodicity, with a nearest-neighbor stencil ($T=\{-1,0,1\}$) highlighted around center point $j=0$. The solid arrow illustrates the cyclic shift $S_{+1}$. (b) 2-dimensional grid on $N=4=2^2$ points per dimension, with nearest-neighbor stencil around center point $(1,1)$. Dashed arrows indicate the periodic wrap-around across grid boundaries. Connecting the boundaries results in a toroid.}
\label{fig: stencil}
\end{figure*}

\begin{definition}[Difference-of-Gaussian Stencil]\label{def:dogstencil}
Let $T \subset \mathbb{Z}^D$ be a finite, truncated symmetric set of spatial offsets. Elements $t \in T$ represent relative coordinate shifts from a center point. For instance, $t=0$ denotes the central pixel, while $t=(\pm 1, \dots, 0)$ denotes its immediate neighbors in the first dimension. We define the discrete DoG operator $A_h$ acting on the grid as

\begin{equation}
    A_h = \sum_{t \in T} c_t S_t,
\end{equation}
where $S_t \ket{j} = \ket{j + t \pmod N}$ is the cyclic spatial shift operator. The coefficients are a result of the difference of two discrete Gaussian mixtures
\begin{equation}
    c_t = p_t - q_t.
\end{equation}

The terms $p_t$ and $q_t$ represent the discrete sampling obtained by evaluating continuous Gaussian densities at the lattice points $t \in T$ and renormalizing so that each forms a valid probability distribution over $T$. Concretely,
\begin{equation}\label{eq:gaussian-kernels}
    p_t \;\propto\; \exp\!\bigl(-\tfrac{\|t\|^2}{2\sigma_p^2}\bigr),
    \qquad
    q_t \;\propto\; \exp\!\bigl(-\tfrac{\|t\|^2}{2\sigma_q^2}\bigr),
\end{equation}
with $\sigma_p < \sigma_q$. The narrow kernel $p$ (small variance $\sigma_p^2$) responds primarily to high-frequency spatial features, while the wide kernel $q$ (large variance $\sigma_q^2$) captures a smoothed local average. Their difference $c_t = p_t - q_t$ suppresses both low-frequency background and high-frequency noise, retaining only features at the intermediate spatial scale defined by $\sigma_p$ and $\sigma_q$, which is the characteristic behavior of a bandpass filter. The spectral analysis in Section \ref{spectral} provides a more detailed analysis of this filtering action.

\end{definition}

We enforce discrete normalization after the truncation such that $p_t, q_t \geq 0$ for all $t \in T$ and $\sum_{t \in T} p_t = \sum_{t \in T} q_t = 1$. Since $p_t$ and $q_t$ are both normalized probability distributions, the $1$-norm of the coefficient vector is naturally bounded by the triangle inequality $\sum_{t \in T} |c_t| \le 2$. This bound is significant for the LCU construction, as it directly dictates the subnormalization factor $\lambda=2$ required for the block-encoding identity derived in the subsequent section.

\section{Circuit Implementation}\label{sec:lcu_implementation}

The Linear Combination of Unitaries (LCU) framework \cite{childs2012hamiltonian} is used to design the circuit for the DoG operator,
\begin{equation}\label{eq: Ah-def}
    A_h = \sum_{t \in T} (p_t - q_t) S_t.
\end{equation}
Here, $\{p_t\}$, $\{q_t\}$, and $S_t$ are defined in Definition~\ref{def:dogstencil}. A direct implementation of Eq. \eqref{eq: Ah-def} typically requires preparing a state with signed amplitudes. For small kernels, we can adopt the state preparation method in \cite{plesch2011quantum}, which is modular with our down-the-line methods. However, we note that such arbitrary-state preparation methods are computationally expensive; specifically, \cite{plesch2011quantum} requires CNOT gates that scale exponentially with the number of qubits for preparing the state. Therefore, for larger stencils, we avoid this extensive computational cost by encoding the sign structure into a single auxiliary qubit. This allows us to inherit the state preparation proposed in \cite{kuklinski2025simpler} for loading the magnitude values up to a normalization factor. The coherent subtraction between $p$ and $q$ can then be realized by using a single Pauli-$Z$ gate on the indicator qubit, which flips the relative phase between the two Gaussian branches.

\subsection{Registers}\label{subsec:registers}

The DoG circuit comprises of an indicator register $(\mathcal{H}_{\text{ind}})$, a shift-label register ($\mathcal{H}_{\text{shift}}$), and a data register ($\mathcal{H}_{\text{data}}$). $\mathcal{H}_{\text{ind}}$ denotes a single qubit used to branch between the two Gaussian components $p$ and $q$, and $\mathcal{H}_{\text{shift}}$ consists of $s$ qubits, where $s= \lceil \log_2 |T| \rceil$. 
Finally, $\mathcal{H}_{\text{data}}$ is a register of $D \times n$ qubits encoding the grid state $\ket{j} = \bigotimes_{k=1}^D \ket{j_k}$ for $j_k \in \mathbb{Z}_N$ and $n = \log_2 N$.

The fundamental unitaries in the construction are the shift operators and the Gaussian loaders. The shift unitary $S_t$ acts on the data register as $S_t \ket{j} = \ket{j + t \pmod{N}}$, where the addition is performed component-wise in $\mathbb{Z}_N^D$. As per our earlier discussion regarding state preparation, we assume the existence of methods $G_p$ and $G_q$ that prepare the magnitude encoded quantum states
\begin{equation}\label{eq:gauss-loaders}
    G_\pi \ket{0}_{\mathrm{shift}}
    = \sum_{t \in T} \sqrt{\pi_t}\, \ket{t}_{\mathrm{shift}},
    \quad \pi \in \{p, q\}.
\end{equation}
These states, along with their adjoints $G_p^\dagger$ and $G_q^\dagger$, constitute the primary non-Clifford cost of the block encoding.

\subsection{Prepare and Shift}\label{subsec:prep-select}

The preparation unitary $P$ acts on the auxiliary space $\mathcal{H}_{\mathrm{ind}} \otimes \mathcal{H}_{\mathrm{shift}}$ and is defined by
\begin{equation}\label{eq:P-def}
    P = \bigl(
        \ketbra{0}{0}_{\mathrm{ind}} \otimes G_p
        + \ketbra{1}{1}_{\mathrm{ind}} \otimes G_q
    \bigr)\,
    (H_{\mathrm{ind}} \otimes I_{\mathrm{shift}}).
\end{equation}
A Hadamard gate on the indicator qubit creates an equal superposition of the two branches, after which the controlled Gaussian loaders prepare the corresponding amplitude distributions. Acting on the all-zeros ancilla state, this yields
\begin{align}\label{eq:P-action}
    P\ket{0,0}_{\mathrm{ind,shift}}
    = \frac{1}{\sqrt{2}}
    \Bigl(
        \ket{0}_{\mathrm{ind}} \sum_{t \in T} \sqrt{p_t}\, \ket{t}_{\mathrm{shift}}
        + \nonumber \\ \ket{1}_{\mathrm{ind}} \sum_{t \in T} \sqrt{q_t}\, \ket{t}_{\mathrm{shift}}
    \Bigr).
\end{align}

The select unitary acts on all three registers and applies the shift $S_t$ to the data register, controlled by the shift register
\begin{equation}\label{eq:SELECT-def}
    \mathrm{SEL}
    = \sum_{t \in T}
      I_{\mathrm{ind}} \otimes \ketbra{t}{t}_{\mathrm{shift}} \otimes S_t.
\end{equation}

\subsection{Block Encoding}\label{subsec:block-encoding}

The full block-encoding unitary, $U$ is
\begin{equation}\label{eq:U-def}
    U = (P^\dagger \otimes I_{\mathrm{data}})\;
        \mathrm{SEL}\;
        (Z_{\mathrm{ind}} \otimes I_{\mathrm{shift}} \otimes I_{\mathrm{data}})\;
        (P \otimes I_{\mathrm{data}}),
\end{equation}
where $Z_{\mathrm{ind}}$ is the Pauli-$Z$ gate on the indicator qubit. This gate induces a relative $\pi$-phase between the $\ket{0}_{\mathrm{ind}}$ and $\ket{1}_{\mathrm{ind}}$ branches, effectively converting the sum of two Gaussian components into their difference. We drop subscripts for the registers whenever they are obvious for compactness.

\begin{proposition}[DoG Block Encoding]\label{prop:block-encoding}
Let $A_h$ be the operator defined in equation \eqref{eq: Ah-def} with $\sum_{t} p_t = \sum_{t} q_t = 1$. Let $U$ be the unitary in \eqref{eq:U-def}. Assuming access to Gaussian state prepare routines ($G_p,\, G_q$ satisfying \eqref{eq:gauss-loaders}), $U$ is a $(\lambda, a, 0)$-block-encoding of $A_h$ with subnormalization $\lambda = 2$, $a = s+1$ ancilla qubits, and zero approximation error
\begin{equation}\label{eq:block-encoding-identity}
    (\bra{0,0} \otimes I_{\mathrm{data}})\, U\, (\ket{0,0} \otimes I_{\mathrm{data}})
    = \frac{1}{2}\, A_h.
\end{equation}
However, if the Gaussian loaders are implemented to finite precision $\epsilon_G$ (i.e., $\|(G_\pi - G_{\tilde{\pi}})\ket{0}\| \leq \epsilon_G$ for $\pi \in \{p,q\}$), they produce a block encoding, $\tilde{U}$ such that,
\begin{equation}
\tilde{U} = (\tilde{P}^\dagger \otimes I)\ \mathrm{SEL}\ (Z_{\mathrm{ind}} \otimes I \otimes I)\ (\tilde{P} \otimes I_{\mathrm{data}}), 
\end{equation}
where $\tilde{P}$ arises from $G_{\tilde{\pi}}$. $\tilde{U}$ is then a $(\lambda, a, \epsilon)$-block-encoding with $\epsilon \leq 2\epsilon_G$.

\end{proposition}

\begin{proof}
Fix an arbitrary data state $\ket{v} \in \mathcal{H}_{\mathrm{data}}$ and consider the initial state $\ket{0}_{\mathrm{ind}}\ket{0}_{\mathrm{shift}}\ket{v}_{\mathrm{data}}$.

\textit{(PREP)}
Applying $P \otimes I_{\mathrm{data}}$ yields
\begin{equation}\label{eq:proof-step1}
    (P \otimes I)\ket{0,0}\ket{v}
    = \frac{1}{\sqrt{2}}
    \Bigl(
        \ket{0}\sum_{t \in T} \sqrt{p_t}\,\ket{t}
        + \ket{1}\sum_{t \in T} \sqrt{q_t}\,\ket{t}
    \Bigr)\ket{v}.
\end{equation}

\textit{(Phase flip)}
The $Z$ gate on the indicator qubit maps $\ket{0} \mapsto \ket{0}$ and $\ket{1} \mapsto -\ket{1}$, giving 
\begin{align}\label{eq:proof-step2}
    (Z_{\mathrm{ind}} \otimes I \otimes I)\,(P \otimes I)\ket{0,0}\ket{v}
    = \nonumber \\ \frac{1}{\sqrt{2}}
    \Bigl(
        \ket{0}\sum_{t} \sqrt{p_t}\,\ket{t}
        - \ket{1}\sum_{t} \sqrt{q_t}\,\ket{t}
    \Bigr)\ket{v}.
\end{align}

\textit{(SELECT)}
The select unitary maps $\ket{t}_{\mathrm{shift}}\ket{v} \mapsto \ket{t}_{\mathrm{shift}} S_t\ket{v}$, producing 
\begin{align}\label{eq:proof-step3}
    \mathrm{SEL}\,(Z \otimes I \otimes I)(P \otimes I)\ket{0,0}\ket{v}
    = \nonumber\\ \frac{1}{\sqrt{2}}
    \Bigl(
        \ket{0}\sum_{t} \sqrt{p_t}\,\ket{t}\,S_t\ket{v}
        - \ket{1}\sum_{t} \sqrt{q_t}\,\ket{t}\,S_t\ket{v}
    \Bigr).
\end{align}

\textit{(Uncompute \& Project)}
We apply $P^\dagger \otimes I_{\mathrm{data}}$ and project onto $\bra{0}_{\mathrm{ind}}\bra{0}_{\mathrm{shift}}$. Decomposing $P^\dagger$ as
\begin{equation}\label{eq:Pdagger-decomp}
    P^\dagger
    = (H_{\mathrm{ind}} \otimes I_{\mathrm{shift}})\,
      \bigl(
          \ketbra{0}{0}_{\mathrm{ind}} \otimes G_p^\dagger
          + \ketbra{1}{1}_{\mathrm{ind}} \otimes G_q^\dagger
      \bigr),
\end{equation}
we first apply the controlled inverses. On the $\ket{0}_{\mathrm{ind}}$ branch, $G_p^\dagger$ acts on the shift register, producing the overlap $\bra{0}_{\mathrm{shift}} G_p^\dagger \ket{t}_{\mathrm{shift}} = \sqrt{p_t}$. Similarly, on the $\ket{1}_{\mathrm{ind}}$ branch, $G_q^\dagger$ produces $\bra{0}_{\mathrm{shift}} G_q^\dagger \ket{t}_{\mathrm{shift}} = \sqrt{q_t}$. After projecting the shift register onto $\bra{0}_{\mathrm{shift}}$, the state on the indicator qubit (for each shift component $t$) is proportional to
\[
    \sqrt{p_t}\,\ket{0}_{\mathrm{ind}}\,(\sqrt{p_t}\, S_t\ket{v})
    \;-\;
    \sqrt{q_t}\,\ket{1}_{\mathrm{ind}}\,(\sqrt{q_t}\, S_t\ket{v}).
\]
The Hadamard gate on the indicator qubit then maps $\ket{0} \mapsto \tfrac{1}{\sqrt{2}}(\ket{0}+\ket{1})$ and $\ket{1} \mapsto \tfrac{1}{\sqrt{2}}(\ket{0}-\ket{1})$. Projecting onto $\bra{0}_{\mathrm{ind}}$ 
yields 
\begin{align}\label{eq:proof-step4-borne}
    (\bra{0,0} \otimes I)\, U\, (\ket{0,0} \otimes I)\ket{v}
    &= \frac{1}{2} \sum_{t \in T}
       \bigl(p_t - q_t\bigr)\, S_t\ket{v} \nonumber\\
    &= \frac{1}{2}\, A_h\ket{v},
\end{align}
establishing \eqref{eq:block-encoding-identity} with subnormalization factor $\lambda = 2$.

For the finite-precision bound, let $\tilde{P}$ and $\tilde{U}$ denote the preparation unitary and full circuit constructed from approximate loaders $\tilde{G}_p, \tilde{G}_q$. We bound the block-encoding error via the triangle inequality
\begin{align}\label{eq:error-bound}
   & \bigl\|(\bra{00}\otimes I)(\tilde{U}-U)(\ket{00}\otimes I)\bigr\|
    \leq \nonumber \\ & \bigl\|(\bra{00}(\tilde{P}^\dagger - P^\dagger)\otimes I)\,\mathrm{SEL}\,\tilde{\Psi}\bigr\| \nonumber\\
   \quad &+ \bigl\|(\bra{00} P^\dagger\otimes I)\,\mathrm{SEL}\,(\tilde{\Psi}-\Psi)\bigr\|,
\end{align}
where $\Psi = (Z\otimes I)(P\otimes I)\ket{00}\ket{v}$ and $\tilde{\Psi}$ is defined analogously with $\tilde{P}$. For the second term, $\|\tilde{\Psi}-\Psi\| = \|(\tilde{P}-P)\ket{00}\|\leq \epsilon_G$, since expanding via~\eqref{eq:P-def} gives
\[
    (\tilde{P}-P)\ket{00} = \frac{1}{\sqrt{2}}\bigl(\ket{0}\otimes(\tilde{G}_p - G_p)\ket{0} + \ket{1}\otimes(\tilde{G}_q - G_q)\ket{0}\bigr),
\]
whose norm satisfies $\frac{1}{2}(\epsilon_G^2 + \epsilon_G^2) \leq \epsilon_G^2$. For the first term, $\bra{00}(\tilde{P}^\dagger - P^\dagger) = \bigl((\tilde{P}-P)\ket{00}\bigr)^\dagger$, so its operator norm is also bounded by $\epsilon_G$. Since $\mathrm{SEL}$ is unitary, both terms contribute at most $\epsilon_G$, yielding $\epsilon \leq 2\epsilon_G$.
\end{proof}

The overall quantum circuit is illustrated in Figure~\ref{fig:total-block}.

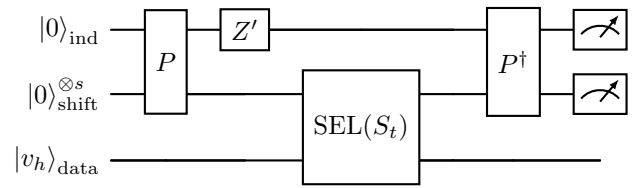
\begin{figure}[ht]
  \centering
  \begin{quantikz}[row sep=0.25cm, column sep=0.45cm]
    \lstick{$\ket{0}_{\mathrm{ind}}$}
      & \gate[wires=2]{P}
      & \gate{Z'}
      & \qw
      & \qw
      & \gate[wires=2]{P^\dagger}
      & \meter{} \\
    \lstick{$\ket{0}^{\otimes s}_{\mathrm{shift}}$}
      &
      & \qw
      & \gate[wires=2]{\mathrm{SEL}(S_t)}
      &
      &
      & \meter{} \\
    \lstick{$\ket{v_h}_{\mathrm{data}}$}
      & \qw
      & \qw
      &
      & \qw
      & \qw
      & \qw
  \end{quantikz}
  \caption{DoG block-encoding circuit. Postselecting the indicator and shift-label registers on $\ket{0}_{\mathrm{ind}}\ket{0}^{\otimes s}_{\mathrm{shift}}$ implements $\widetilde{A}_h = A_h/2$ on the data register. Here, $Z'$ denotes the $\mathrm{Pauli}-Z$ gate applied on the indicator register to produce the phase shift between the two branches.}
  \label{fig:total-block}
\end{figure}

\subsection{Resource Estimates}
As evident from Figure \ref{fig:total-block} and Section \ref{sec:lcu_implementation}, the block encoding circuit $U$ consists of three Clifford gates (two Hadamards and a Pauli-$Z$), the state preparation block $P$, and the controlled shift operators $S_t$. We consider the non-Clifford cost contained within each of the Gaussian loaders $G_p$ and $G_q$ (alongside their adjoints) in $P$, and the explicit circuit needed to realize the controlled shift unitaries in $\mathrm{SEL}$.

Each Gaussian loader acts on $s=\lceil log_2(T) \rceil$ qubits to prepare the amplitude encoded state, which, by the use of standard state preparation method in \cite{plesch2011quantum}, allows implementation with a cost of $\mathcal{O}(2^s)$ rotations with no required ancilla qubits. This method is fully compatible with our proposed block-encoding circuit, and the cost is acceptable for smaller Gaussian kernels used in image processing and edge-detection applications. In cases where the kernel size is significantly large as to make the cost scaling of $\mathcal{O}(2^s)$ prohibitive, the Gaussian structure-aware method in \cite{kuklinski2025simpler} reduces the cost to $\mathcal{O}(s^2) $ rotations using $\lfloor (s-1)/2\rfloor$ additional ancillary qubits. However, we observe that this method, in turn, depends on the ``post-selection-by-measurement" method, which introduces an overall multiplicative constant on our block-encoding subnormalization factor $\lambda$. Specifically, the ancilla introduces a subnormalization, $\gamma<1$, within the loader $P$, which results in our overall block encoding to have an effective subnormalization of $\lambda_{\mathrm{eff}} = 2/\gamma$. This new $\lambda_{\mathrm{eff}}$ has an adverse effect on our final post-selection probability, and therefore, trade-offs between circuit cost and post-selection probability must be considered, given the stencil size used in the application. 

Next, we observe that the $\mathrm{SEL}$ block applies a $D$-dimensional shift $S_t$ via $D$ constant addition circuits on $n = log_2(N)$ qubits, which has a cost of $\mathcal{O}(n)$ Toffoli gates per addition \cite{draper2004logarithmic} for a total of $\mathcal{O}(|T| \cdot D \cdot \log N)$ T-gates. Each rotation can be synthesized to precision $\epsilon_G$ at a cost of $\mathcal{O}(\log(1/\epsilon_G))$ T-gates \cite{bocharov2015efficient}. We finally consider all the non-Clifford costs to arrive at $\mathcal{O}(2^s \log(1/\epsilon_G) + |T| \cdot D \cdot \log N)$ T-gates with subnormalization $\lambda = 2$ for the structure-agnostic loaders, or $\mathcal{O}(s^2 \log(1/\epsilon_G) + |T| \cdot D \cdot \log N)$ with $\lambda_{\mathrm{eff}} = 2/\gamma$ when adopting the structure-exploiting prepare methods.

\section{Spectral Characterization of the DoG Operator}\label{spectral}

The proposed block encoding in Section \ref{sec:lcu_implementation} can be implemented 
as the projection action on a linear combination of shift unitaries. In this section, we derive the spectral characterization of the DoG operator, which provides the necessary tools to analyze the operator's norm, Hermicity, and the encoding's success probability. Since the spatial domain is periodic, the discrete shift operators $S_t$ are circulant on $\mathbb{Z}_N^D$ and are therefore diagonalized by the $D$-dimensional discrete Fourier basis $\{\ket{\omega}\}_{\omega \in \mathbb{Z}_N^D}$, where

\[
\ket{\omega}
=
\frac{1}{\sqrt{N^D}}
\sum_{j\in\mathbb{Z}_N^D}
e^{2\pi i \langle \omega, j\rangle / N}\ket{j}.
\]

\begin{theorem}[Operator Norm of the Periodic DoG Operator] \label{theorem: spectral}
Let
\[
A_h = \sum_{t \in T} (p_t - q_t) S_t
\]
be the periodic Difference-of-Gaussian operator on $\mathbb{Z}_N^D$ and let $\widehat{p}(\omega)$ and $\widehat{q}(\omega)$ denote the discrete Fourier transforms of the coefficient sequences $\{p_t\}$ and $\{q_t\}$:
\[
\widehat{p}(\omega) = \sum_{t \in T} p_t e^{-2\pi i \langle \omega, t \rangle / N},
\qquad
\widehat{q}(\omega) = \sum_{t \in T} q_t e^{-2\pi i \langle \omega, t \rangle / N}.
\]
Then $A_h$ is diagonal in the discrete Fourier basis, with eigenvalues
\[
\mu(\omega)
=
\widehat{p}(\omega)
-
\widehat{q}(\omega).
\]
Furthermore, if the kernels are symmetric in the sense that $p_t=p_{-t}$ and $q_t=q_{-t}$ (equivalently, $c_t:=p_t-q_t$ satisfies $c_t=c_{-t}$), then $A_h$ is Hermitian and the eigenvalues $\mu(\omega)$ are real. Its operator norm satisfies
\[
\|A_h\|
=
\max_{\omega \in \mathbb{Z}_N^D}
\left|
\widehat{p}(\omega)
-
\widehat{q}(\omega)
\right|
=
\max_{\omega \in \mathbb{Z}_N^D}
|\mu(\omega)|.
\]
\label{frequency domain}
\end{theorem}

\begin{proof}
Applying the spatial shift operator $S_t$ to a Fourier basis state $\ket{\omega}$ yields
\[
S_t \ket{\omega}
=
e^{-2\pi i \langle \omega, t \rangle / N}
\ket{\omega}.
\]
Thus, the action of the full operator evaluates to
\begin{align}
A_h \ket{\omega}
=
\left(
\sum_{t \in T} (p_t - q_t)
e^{-2\pi i \langle \omega, t \rangle / N}
\right)
\ket{\omega}
\\=
\bigl(\widehat{p}(\omega)-\widehat{q}(\omega)\bigr)\ket{\omega},
\end{align}
so the Fourier basis diagonalizes $A_h$ with eigenvalues $\mu(\omega)=\widehat{p}(\omega)-\widehat{q}(\omega)$.
If $c_t=c_{-t}\in\mathbb{R}$, then using $S_t^\dagger=S_{-t}$ we have
\[
A_h^\dagger
=
\sum_{t\in T} c_t S_t^\dagger
=
\sum_{t\in T} c_t S_{-t}
=
\sum_{t\in T} c_{-t} S_t
=
A_h,
\]
so $A_h$ is Hermitian and hence normal. Therefore, its operator norm equals the maximum absolute eigenvalue, giving the stated result.
\end{proof}

\subsection{Algorithmic Consequences}

Theorem \ref{frequency domain} makes the spectrum of the periodic DoG operator explicit and provides an exact characterization of its operator norm. Together with our explicit block encoding, this enables direct use of $A_h$ in eigenvalue-based quantum primitives and clarifies its action as a bandpass filter in the Fourier basis.

\subsubsection{Direct Hamiltonian Simulation and Eigenvalue Transformations}
Hamiltonian simulation and eigenvalue-based transformations require Hermitian generators. For a general non-Hermitian operator $A$, one commonly proceeds via a Hermitian dilation (e.g., $\begin{pmatrix} 0 & A \\ A^\dagger & 0 \end{pmatrix}$), which increases dimension and introduces additional overhead. In contrast, under the symmetry condition $c_t=c_{-t}\in\mathbb{R}$, the DoG stencil $A_h$ is Hermitian with real spectrum. Consequently, given a block encoding of $A_h$, one may apply qubitization/quantum signal processing (QSP)-style eigenvalue transformations directly to implement functions of $A_h$, including time-evolution operators $e^{iA_h t}$, without first forming a dilation.

\subsubsection{Bandpass Structure in the Fourier Basis}

The eigenvalue response $\mu(\omega) = \widehat{p}(\omega) - \widehat{q}(\omega)$ characterizes the exact frequency-domain action of $A_h$. Normalization of the kernels implies $\widehat{p}(0) = \widehat{q}(0) = 1$, so $\mu(0) = 0$: the operator annihilates the DC component (response at $\omega = 0$). For Gaussian-like kernels, $\widehat{p}(\omega)$ and $\widehat{q}(\omega)$ are suppressed at high frequencies, so $\mu(\omega)$ is concentrated on an intermediate band controlled by $\sigma_p$ and $\sigma_q$. In the following section, we use this eigenvalue decomposition to derive an exact expression for the block-encoding success probability and to establish its asymptotic scaling.

\section{Success Probability Analysis}\label{sec: success_prob}

In this section, we describe the success probability of our proposed method from two different perspectives. We start by deriving an expression for any input state at a given, fixed grid resolution, quantifying the filtering action of the DoG operator. Next, we establish that the success probability scales as $O(h^4)$ in the asymptotic regime where the grid spacing $h \to 0$ for a smooth function.

\subsection{Closed form Success Probability} \label{subsec: closed_suc_prob}

The following proposition provides a closed-form representation of the success probability of the Block Encoding and follows directly from the spectral analysis of the DoG operator described in Theorem \ref{theorem: spectral} of Section \ref{spectral}. 

\begin{proposition}[Success Probability: Closed-Form]\label{prop: suc_prob_parseval} Let $ \ket{v_h} \in \mathcal{H}_{data}$ be an arbitrary quantum state with Fourier coefficients $\hat{v}_h(\omega) = \braket{\omega}{v_h}$ where $\ket{\omega}$ is the Fourier basis state. Then, the exact success probability of the proposed block encoding U in Proposition \ref{prop:block-encoding} is 

\begin{equation}\label{eq:parseval-success}
    P_{\mathrm{success}}
    = \left\| \frac{1}{2} A_h \ket{v_h} \right\|^2
    = \frac{1}{4} \sum_{\omega \in \mathbb{Z}_N^D}
      |\mu(\omega)|^2 \, |\hat{v}_h(\omega)|^2,
\end{equation} 
where $\mu(\omega) = \hat{p}(\omega) - \hat{q}(\omega)$ are the eigenvalues of $A_h$ given by Theorem \ref{theorem: spectral}.
\end{proposition}

\begin{proof}
The first equality follows directly from applying Born's rule on Equation \eqref{eq:proof-step4-borne}. We showed that $A_h$ is diagonal in the Fourier basis in Theorem \ref{theorem: spectral}, which implies the action of $A_h$ in the Fourier Basis is $A_h \ket{\omega} = \mu(\omega)\ket{\omega}$. Expanding $\ket{v_h} = \sum_{\omega} \hat{v}_h(\omega)\ket{\omega}$ gives
\begin{align}
    \frac{1}{\lambda} A_h \ket{v_h}
    &= \frac{1}{2} \sum_{\omega \in \mathbb{Z}_N^D} A_{h} \hat{v}_h(\omega)\ket{\omega}\\
    &= \frac{1}{2} \sum_{\omega \in \mathbb{Z}_N^D}
      \mu(\omega)\, \hat{v}_h(\omega) \ket{\omega}.
\end{align}
By Parseval's theorem (equivalently, unitarity of the discrete Fourier transform) and the orthonormality of the Fourier basis,
\begin{align}
    P_{\mathrm{success}}
    = \frac{1}{4} \sum_{w, w^{*} \in \mathbb{Z}_N^D} \mu(\omega)\mu(\omega^*)\hat{v_h}(\omega)\hat{v_h}(\omega^*)\braket{\omega}{\omega^*}\\
    = \frac{1}{4} \sum_{\omega \in \mathbb{Z}_N^D}
      |\mu(\omega)|^2 \, |\hat{v}_h(\omega)|^2
\end{align}  
where $\omega^*$ is the complex conjugate of $\omega$. \qedhere
\end{proof}

Equation \eqref{eq:parseval-success} has a direct filtering interpretation. It shows that the success probability is one-quarter of the weighted average of the squared transfer function, $|\mu(\omega)|^2$, with weights derived from $|\hat{v}_h(\omega)|^2$, which is the power spectral density of the input signal. Normalization of $\ket{v_h}$ guarantees $\sum_\omega |\hat{v}_h(\omega)|^2 = 1$, and therefore
\begin{equation}\label{eq:success-bound-spectral}
    P_{\mathrm{success}}
    \leq \frac{1}{4} \max_{\omega} |\mu(\omega)|^2
    = \frac{1}{4} \| A_h \|^2.
\end{equation}
where the first inequality follows from the weighted average bound, and finally the third expression is equal as a direct consequence of $A_h$ being Hermitian, as proven in Theorem \ref{theorem: spectral}. Since both $p$ and $q$ are normalized, the response at $\omega = 0$ vanishes ($\mu(0) = 0$). Gaussian-like kernels are essentially band-pass filters, which means that $|\mu(\omega)|$ is also suppressed at high frequencies, so the success probability is concentrated on input energy lying within the DoG passband. In particular, the success probability is high when the signal has significant spectral energy in the intermediate frequency range defined by $\sigma_p$ and $\sigma_q$, and low when the signal is concentrated at DC or at frequencies beyond the filter's response.

\subsection{Asymptotic Scaling for Smooth Functions} \label{subsec:asymp_smooth}
For the case in which the input is taken by sampling from a progressively refined grid, $h \to 0$, the asymptotic success probability grows in the order of $O(h^4)$. This asymptotic characterization connects the DoG operator to the continuous Laplacian. We formalize the observation by the following theorem.

\begin{theorem}[Asymptotic Success Probability]\label{theorem: success-asymptotic}
Let $v \in \mathcal{C}_{per}^4(\Omega_D)$, the space of four-times continuously differentiable functions on $\Omega_D$, and let $\ket{v_h}$ be the normalized quantum state obtained by sampling $v$ on the periodic grid $\mathbb{Z}_N^D$ as in Definition \ref{def: grid}. Then
\begin{equation}\label{eq:success-asymptotic}
    P_{\mathrm{success}}
    = \frac{C_{\mathrm{DoG}}^2}{4D^2}\, h^4\,
      \frac{\|\Delta v\|_{L^2(\Omega_D)}^2}
           {\|v\|_{L^2(\Omega_D)}^2}
      + \mathcal{O}(h^6),
\end{equation}
where $C_{\mathrm{DoG}} = \frac{1}{2} \sum_{t \in T} (p_t - q_t)\|t\|^2$, $\Delta$ denotes the continuous Laplacian on $\Omega_D$ and $\|.\|_{L^2(\Omega_D)}$ denotes the continuous $L^2$ norm defined in $D-dimensional$ $\Omega_D$. 
\end{theorem}

\begin{proof}

We expand the action of $A_h$ on $v$ using the Taylor expansion as a continuous point $x \in \Omega_D$, 
\begin{align}\label{eq:taylor}
    (A_h v)(x)  = \sum_{t \in T} (p_t - q_t)  \biggl( v(x) + 
    h\, t \cdot \nabla v(x) \nonumber \\ + \frac{h^2}{2}\, t^T H(v(x))\, t + \mathcal{O}(h^3)\biggr)
\end{align}
where $H(v(x))$ is the Hessian matrix.

The constant term vanishes due to normalization, $\sum_t (p_t - q_t) = 0$, the linear term, and all odd-order terms vanish by kernel symmetry ($p_t = p_{-t}$, $q_t = q_{-t}$). Therefore, the quadratic term is the first non-zero term. 
We know that Gaussian kernels are isotropic, and therefore, the weighted quadratic form is rotationally symmetric, so
\[
\sum_t (p_t-q_t)\, t t^\top = \frac{C_{\mathrm{DoG}}}{D}\, I.
\]
This shows that off-diagonal Hessian terms cancel and therefore, the quadratic Taylor term satisfies the following simplification
\begin{align*}
\frac{h^2}{2}\sum_t (p_t-q_t)\, t^\top \nabla^2 v(x)\, t
= \\
\frac{h^2}{2}\operatorname{tr}\!\left(\Big(\sum_t (p_t-q_t)\, t t^\top\Big)\nabla^2 v(x)\right) \\
=
\frac{C_{\mathrm{DoG}}}{D}\, h^2\, \Delta v(x),
\end{align*}
giving
\begin{equation}\label{eq:operator-approx}
    (A_h v)(x)
    = \frac{C_{\mathrm{DoG}}}{D}\, h^2\, \Delta v(x)
      + \mathcal{O}(h^4).
\end{equation}

We now pass from the discrete $\ell^2$ norm to the continuous $L^2$ norm via the Riemann sum equivalence $\|w_h\|_2^2 \approx h^{-D} \|w\|_{L^2(\Omega_D)}^2$ (\cite{sturm2025efficient}). Applying this to both the numerator and denominator of the success probability from Proposition~\ref{prop: suc_prob_parseval},
\begin{align}
    P_{\mathrm{success}}
    &= \frac{1}{4}\,
       \frac{\|A_h v_h\|_2^2}{\|v_h\|_2^2} =
       \frac{1}{4}\,
       \frac{h^{-D} \left\| \frac{C_{\mathrm{DoG}}}{D}\, h^2\, \Delta v \right\|_{L^2}^2}
            {h^{-D}\, \|v\|_{L^2}^2} + \mathcal{O}(h^6)
    \nonumber\\[4pt]
    &\approx \frac{C_{\mathrm{DoG}}^2}{4D^2}\, h^4\,
       \frac{\|\Delta v\|_{L^2(\Omega_D)}^2}
            {\|v\|_{L^2(\Omega_D)}^2},
\end{align}
where the $h^{-D}$ factors cancel, establishing the $\mathcal{O}(h^4)$ scaling.
\end{proof}

Theorem \ref{theorem: success-asymptotic} shows that our DoG block encoding, in the grid $h \to 0$ regime, has the same success-probability as the Laplacian block encoding in \cite{sturm2025efficient}. However, the block-encoding methods differ significantly. In the construction in \cite{sturm2025efficient}, the operator $L_{D,h}$ has eigenvalues of order $1/h^2$, consequently causing the subnormalization to be $\lambda_{D,\max} = 4D/h^2$, which causes the success-probability to scale as $\mathcal{O}(h^{4})$. In our construction, we design the operator $A_h$, which is a linear combination of shifts, have eigenvalues bounded by a constant ($\|A_h\| \leq 2$), so $\lambda = 2$ is independent of $h$ and instead cause the output amplitude $\|A_h \ket{v_h}\|$ to shrinks as $h^2$ for smooth inputs, producing the same net $\mathcal{O}(h^4)$ scaling upon squaring. We observe that this matching is not coincidental, but rather stems from a fundamental constraint for any operator that approximates a second-order derivative on an $h-$discretized grid. The $h^2$ factor must be accounted for either in subnormalization or in the output amplitude. We push this factor to the amplitude to achieve a constant sub-normalization and observe that this mathematical distinction between the two methods yields a significant advantage in the context of the DoG operator. This advantage does not manifest in the grid-refinement regime, as both methods exhibit the same scaling, but rather lies down the pipeline for eigenvalue transformation that follows the block encoding. Specifically, if we are performing Hamiltonian simulation via $e^{iA_h t}$, the effective simulation time scales as $\lambda t$, which means that our construction requires a time of $2t$ independent of the $h^2$ factor. In a broader perspective, QSVT-like spectral transformations of $A_h$, such as bandpass thresholding or selective amplification of modes within the DoG passband, polynomial degrees, depend only on the filter parameters $\sigma_p$, $\sigma_q$, and are independent of the grid spacing $h$. This is a direct consequence of our eigenvalues, $\mu(\omega)/\lambda$, occupying a fixed $h$-independent fraction of $[-1,1]$.

As for fixed-resolution filtering, which is the setting of our Proposition \ref{prop: suc_prob_parseval}, there is no grid refinement, and therefore, the success probability is governed directly by the signal's energy in the DoG passband.

\section{Numerical Example}\label{sec:numerical}
We illustrate our method with a simple 1-D example. We use the parameters: $D=1$, $N=16$ ($n=4$ qubits), $\sigma_p = 0.8$, $\sigma_q = 1.6$, and stencil $T = \{-3, \ldots, 3\}$ ($s = 3$ shift-label qubits).

Figure \ref{fig:freq_response} displays the two Gaussian distributions, the DoG coefficients, and the eigenvalue response $\mu(\omega) = \hat{p}(\omega) - \hat{q}(\omega)$ over the frequency modes $\omega = 0, \ldots, N/2$. The transfer function confirms the bandpass structure, predicted by Theorem \ref{theorem: spectral}, $\mu(0) = 0$, a peak at intermediate frequencies, and suppression at high frequencies.

Figure \ref{fig:success_prob} compares the exact success probability from Proposition \ref{prop: suc_prob_parseval} with the asymptotic formula of Theorem \ref{theorem: success-asymptotic} for the input $v(x) = \sin(2\pi x)$ across grid sizes $N = 2^4$ through $2^{10}$. Both curves follow the $N=h^{-1}, ~\mathcal{O}(h^{4})$ scaling, and their ratio converges to $1$ as the grid is refined, confirming that the asymptotic approximation is
consistent with the exact Parseval-based expression.

The narrow Gaussian $p$ has a higher concentration of its mass at $t=0$, while the wider $q$ has more weight distributed throughout the stencil width. The difference of the Gaussians illustrates the signed structure, which is positive at the center and negative near the edge of the stencils. The coefficient $1$-norm calculates the $\sum|c_{t}|$ to be $\approx 0.556$, which is below the universal bound $\lambda = 2$ used in the block encoding. Although a tighter normalization could be achieved, it would require us to concede the structured prep-shift structure that makes the circuit explicit and efficient.

\begin{figure*}
    \centering
    \includegraphics[width=1\linewidth]{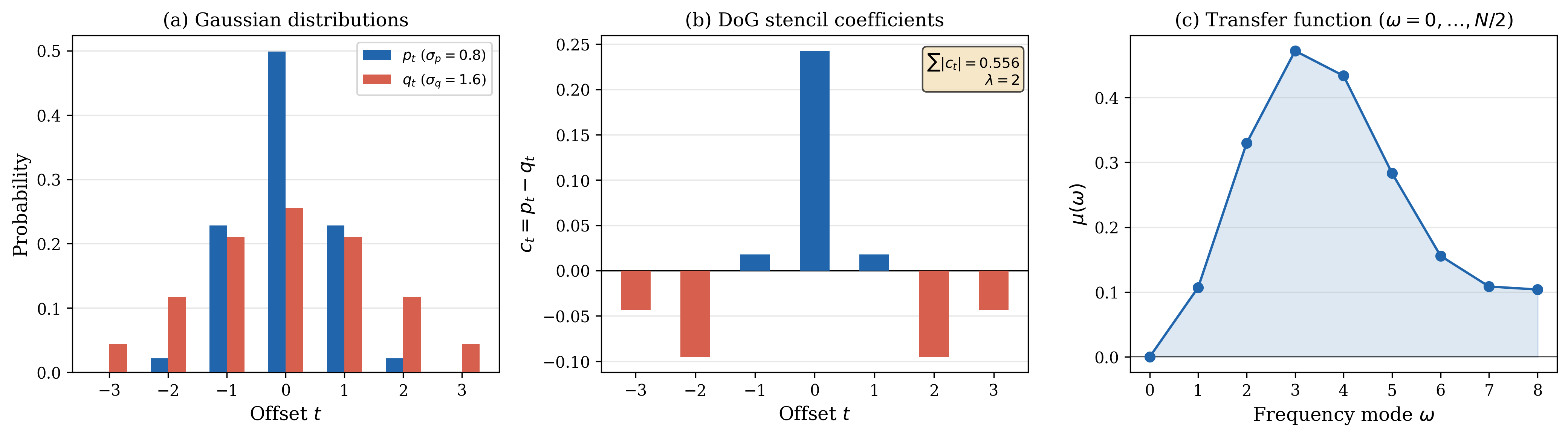}
    \caption{Plots showing the (a) Gaussian kernel shapes, (b) the stencil coefficients, and (c) the transfer function with respect to the $\omega$. The values are for the example with the parameters:  $D=1$, $\sigma_p=0.8$, $\sigma_q=1.6$, $T=\{-3,\ldots,3\}$}
    \label{fig:freq_response}
\end{figure*}

\begin{figure*}
    \centering
    \includegraphics[width=1\linewidth]{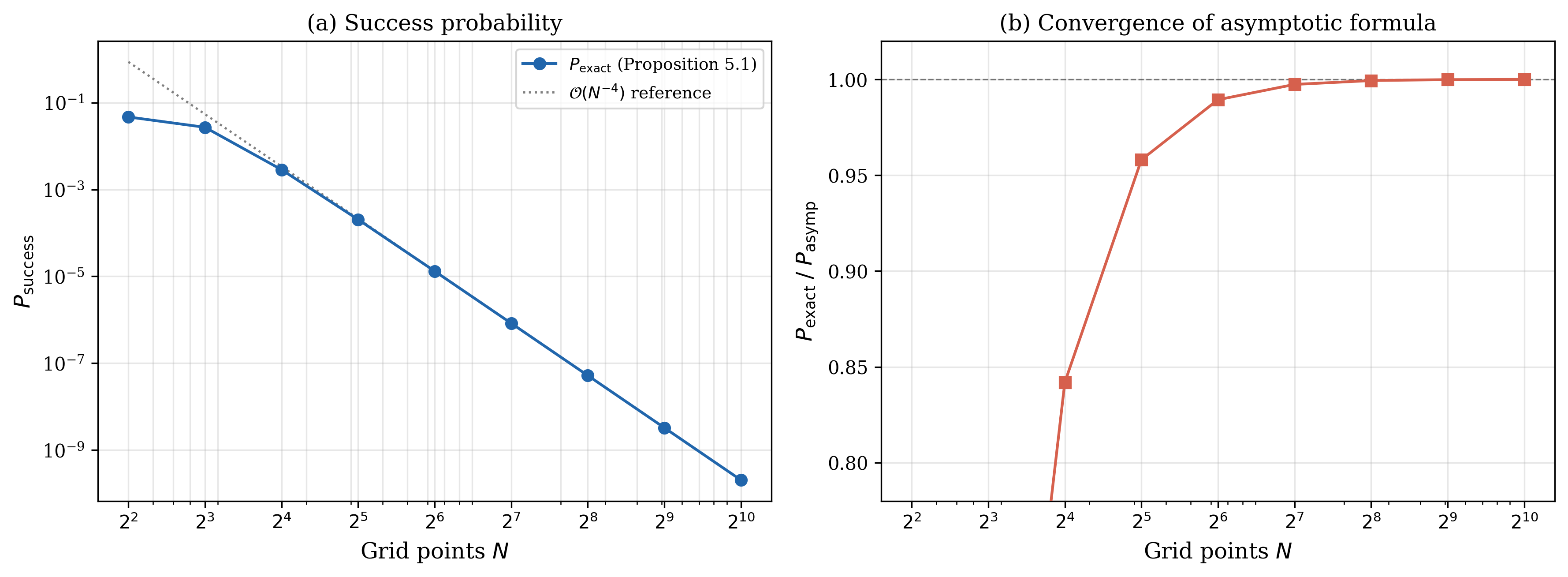}
    \caption{The scaling of success probability with respect to grid points, $N$ where $N=h^{-1}$}
    \label{fig:success_prob}
\end{figure*}

\section{Discussion}
In this work, we develop an explicit block encoding of the DoG operator on a discretized periodic grid, based on the observation that the DoG stencil decomposes as a difference of two normalized discrete Gaussian distributions, which can be mapped to a linear combination of shifts within the LCU framework. The construction achieves an inherent constant subnormalization using a Gaussian state-preparation method, controlled-shift operators, and three single-qubit Clifford gates.

We derive a success probability scaling of $\mathcal{O}(h^4)$ as the grid becomes more refined, which is a fundamental constraint arising from the $h^2$ factor in any second-order differential approximation and therefore, matches the scaling of explicit Laplacian encodings. The constant subnormalization of $\lambda = 2$ achieved in this method realizes benefits when the DoG itself is the target of downstream quantum primitives, for example, in QSVT-based spectral transformation regimes, where it is important to avoid spectral compression. We finally show that for fixed-resolution filtering, the exact success probability depends directly on the signal's spectral overlap with the DoG passband.

The block encoding of the DoG operator can be directly used to implement various applications in quantum computing without the need for black-box oracles or expensive NEQR-based arithmetic circuits used in quantum image processing.

In a broader perspective, we observe that our method applies to any operator expressible as a signed linear combination of normalized distributions over unitary shifts, suggesting potential extensions beyond the Gaussian case. Future work includes resource benchmarking via Qualtran \cite{harrigan2024expressing}, extension to the Dirichlet and Neumann boundary conditions, and numerical experiments applying the DoG block encoding within QSVT-based image processing pipelines.


\section*{Competing interests}
All authors declare no financial or non-financial competing interests.

\bibliographystyle{unsrt}
\bibliography{bib}

@article{kharazi2025explicit,
  title={Explicit block encodings of boundary value problems for many-body elliptic operators},
  author={Kharazi, Tyler and Alkadri, Ahmad M and Liu, Jin-Peng and Mandadapu, Kranthi K and Whaley, K Birgitta},
  journal={Quantum},
  volume={9},
  pages={1764},
  year={2025},
  publisher={Verein zur F{\"o}rderung des Open Access Publizierens in den Quantenwissenschaften}
}

@article{sotak1989laplacian,
  title={The Laplacian-of-Gaussian kernel: a formal analysis and design procedure for fast, accurate convolution and full-frame output},
  author={Sotak Jr, George E and Boyer, Kim L},
  journal={Computer vision, graphics, and image processing},
  volume={48},
  number={2},
  pages={147--189},
  year={1989},
  publisher={Elsevier}
}

@article{neycenssac1993contrast,
  title={Contrast enhancement using the Laplacian-of-a-Gaussian filter},
  author={Neycenssac, Franck},
  journal={CVGIP: Graphical Models and Image Processing},
  volume={55},
  number={6},
  pages={447--463},
  year={1993},
  publisher={Elsevier}
}

@article{kong2013generalized,
  title={A generalized Laplacian of Gaussian filter for blob detection and its applications},
  author={Kong, Hui and Akakin, Hatice Cinar and Sarma, Sanjay E},
  journal={IEEE transactions on cybernetics},
  volume={43},
  number={6},
  pages={1719--1733},
  year={2013},
  publisher={IEEE}
}

@article{polakowski1997computer,
  title={Computer-aided breast cancer detection and diagnosis of masses using difference of Gaussians and derivative-based feature saliency},
  author={Polakowski, William E and Cournoyer, Donald A and Rogers, Steven K and DeSimio, Martin P and Ruck, Dennis W and Hoffmeister, Jeffrey W and Raines, Richard A},
  journal={IEEE transactions on medical imaging},
  volume={16},
  number={6},
  pages={811--819},
  year={1997},
  publisher={IEEE}
}

@inproceedings{assirati2014performing,
  title={Performing edge detection by Difference of Gaussians using q-Gaussian kernels},
  author={Assirati, Lucas and Silva, N{\'u}bia Rosa da and Berton, Lilian and Lopes, Alneu de A and Bruno, Odemir M},
  booktitle={Journal of Physics: Conference Series},
  volume={490},
  number={1},
  pages={012020},
  year={2014},
  organization={IOP Publishing}
}

@article{fan2023quantum,
  title={Quantum image edge extraction based on difference of Gaussian operator.},
  author={Fan, Ping and Xiao, Ke},
  journal={Quantum Information Processing},
  volume={22},
  number={1},
  year={2023}
}

@article{mu2016multiscale,
  title={Multiscale edge fusion for vehicle detection based on difference of Gaussian},
  author={Mu, Kenan and Hui, Fei and Zhao, Xiangmo and Prehofer, Christian},
  journal={Optik},
  volume={127},
  number={11},
  pages={4794--4798},
  year={2016},
  publisher={Elsevier}
}

@article{yao2017quantum,
  title={Quantum image processing and its application to edge detection: Theory and experiment},
  author={Yao, Xi-Wei and Wang, Hengyan and Liao, Zeyang and Chen, Ming-Cheng and Pan, Jian and Li, Jun and Zhang, Kechao and Lin, Xingcheng and Wang, Zhehui and Luo, Zhihuang and others},
  journal={Physical Review X},
  volume={7},
  number={3},
  pages={031041},
  year={2017},
  publisher={APS}
}

@article{wu2022quantum,
  title={Quantum SUSAN edge detection based on double chains quantum genetic algorithm},
  author={Wu, Chenyi and Huang, Fei and Dai, Jingyi and Zhou, Nanrun},
  journal={Physica A: Statistical Mechanics and its Applications},
  volume={605},
  pages={128017},
  year={2022},
  publisher={Elsevier}
}

@article{sundani2019identification,
  title={Identification of image edge using quantum canny edge detection algorithm},
  author={Sundani, Dini and Widiyanto, Sigit and Karyanti, Yuli and Wardani, Dini Tri},
  journal={Journal of ICT research and applications},
  volume={13},
  number={2},
  pages={133--144},
  year={2019}
}

@article{chetia2021quantum,
  title={Quantum image edge detection using improved Sobel mask based on NEQR.},
  author={Chetia, Rajib and Boruah, SMB and Sahu, PP},
  journal={Quantum Information Processing},
  volume={20},
  number={1},
  year={2021}
}

@article{zhou2019quantum,
  title={Quantum image edge extraction based on improved Prewitt operator.},
  author={Zhou, Ri-Gui and Yu, Han and Cheng, Yu and Li, Feng-Xin},
  journal={Quantum Information Processing},
  volume={18},
  number={9},
  year={2019}
}

@article{liu2023quantum,
  title={Quantum image edge detection based on eight-direction Sobel operator for NEQR},
  author={Liu, Wenjie and Wang, Lu},
  journal={arXiv preprint arXiv:2310.03037},
  year={2023}
}

@inproceedings{yuan2019fast,
  title={Fast Laplacian of Gaussian edge detection algorithm for quantum images},
  author={Yuan, Suzhen and Venegas-Andraca, Salvador E and Zhu, Chaoping and Wang, Yan and Mao, Xuefeng and Luo, Yuan},
  booktitle={2019 IEEE International Conferences on Ubiquitous Computing \& Communications (IUCC) and Data Science and Computational Intelligence (DSCI) and Smart Computing, Networking and Services (SmartCNS)},
  pages={798--802},
  year={2019},
  organization={IEEE}
}

@article{li2020quantum,
  title={Quantum implementation of classical Marr--Hildreth edge detection.},
  author={Li, Panchi and Shi, Tong and Lu, Aiping and Wang, Bing},
  journal={Quantum Information Processing},
  volume={19},
  number={2},
  year={2020}
}

@article{childs2012hamiltonian,
  title={Hamiltonian simulation using linear combinations of unitary operations},
  author={Childs, Andrew M and Wiebe, Nathan},
  journal={arXiv preprint arXiv:1202.5822},
  year={2012}
}

@article{yuan2024quantum,
  title={Quantum image edge detection based on Laplacian of Gaussian operator},
  author={Yuan, Suzhen and Zhao, Wenhao and Deng, Jeremiah D and Xia, Shuyin and Li, Xianli},
  journal={Quantum information processing},
  volume={23},
  number={5},
  pages={178},
  year={2024},
  publisher={Springer}
}

@article{sturm2025efficient,
  title={Efficient and Explicit Block Encoding of Finite Difference Discretizations of the Laplacian},
  author={Sturm, Andreas and Schillo, Niclas},
  journal={arXiv preprint arXiv:2509.02429},
  year={2025}
}

@inproceedings{gilyen2019quantum,
  title={Quantum singular value transformation and beyond: exponential improvements for quantum matrix arithmetics},
  author={Gily{\'e}n, Andr{\'a}s and Su, Yuan and Low, Guang Hao and Wiebe, Nathan},
  booktitle={Proceedings of the 51st annual ACM SIGACT symposium on theory of computing},
  pages={193--204},
  year={2019}
}

@article{berry2007efficient,
  title={Efficient quantum algorithms for simulating sparse Hamiltonians},
  author={Berry, Dominic W and Ahokas, Graeme and Cleve, Richard and Sanders, Barry C},
  journal={Communications in Mathematical Physics},
  volume={270},
  number={2},
  pages={359--371},
  year={2007},
  publisher={Springer}
}

@article{childs2017quantum,
  title={Quantum algorithm for systems of linear equations with exponentially improved dependence on precision},
  author={Childs, Andrew M and Kothari, Robin and Somma, Rolando D},
  journal={SIAM Journal on Computing},
  volume={46},
  number={6},
  pages={1920--1950},
  year={2017},
  publisher={SIAM}
}

@article{camps2024explicit,
  title={Explicit quantum circuits for block encodings of certain sparse matrices},
  author={Camps, Daan and Lin, Lin and Van Beeumen, Roel and Yang, Chao},
  journal={SIAM Journal on Matrix Analysis and Applications},
  volume={45},
  number={1},
  pages={801--827},
  year={2024},
  publisher={SIAM}
}

@article{liszka1980finite,
  title={The finite difference method at arbitrary irregular grids and its application in applied mechanics},
  author={Liszka, Tadeusz and Orkisz, Janusz},
  journal={Computers \& Structures},
  volume={11},
  number={1-2},
  pages={83--95},
  year={1980},
  publisher={Elsevier}
}

@book{vuik2023numerical,
  title={Numerical methods for ordinary differential equations},
  author={Vuik, Kees and Vermolen, Fred and van Gijzen, Martin and Vuik, Thea},
  year={2023},
  publisher={TU Delft Open Publishing}
}

@article{mirzaee2022solution,
  title={Solution of time-fractional stochastic nonlinear sine-Gordon equation via finite difference and meshfree techniques},
  author={Mirzaee, Farshid and Rezaei, Shadi and Samadyar, Nasrin},
  journal={Mathematical Methods in the Applied Sciences},
  volume={45},
  number={7},
  pages={3426--3438},
  year={2022},
  publisher={Wiley Online Library}
}

@article{leveque1998finite,
  title={Finite difference methods for differential equations},
  author={LeVeque, Randall J},
  journal={Draft version for use in AMath},
  volume={585},
  number={6},
  pages={112},
  year={1998}
}

@book{smith1985numerical,
  title={Numerical solution of partial differential equations: finite difference methods},
  author={Smith, Gordon D},
  year={1985},
  publisher={Oxford university press}
}

@article{marr1980theory,
  title={Theory of edge detection},
  author={Marr, David and Hildreth, Ellen},
  journal={Proceedings of the Royal Society of London. Series B. Biological Sciences},
  volume={207},
  number={1167},
  pages={187--217},
  year={1980},
  publisher={The Royal Society London}
}

@article{berry2015simulating,
  title={Simulating Hamiltonian dynamics with a truncated Taylor series},
  author={Berry, Dominic W and Childs, Andrew M and Cleve, Richard and Kothari, Robin and Somma, Rolando D},
  journal={Physical review letters},
  volume={114},
  number={9},
  pages={090502},
  year={2015},
  publisher={APS}
}

@article{kuklinski2025efficient,
  title={Efficient block-encodings require structure},
  author={Kuklinski, Parker and Rempfer, Benjamin and Elenewski, Justin and Obenland, Kevin},
  journal={arXiv preprint arXiv:2509.19667},
  year={2025}
}

@article{berry2019qubitization,
  title={Qubitization of arbitrary basis quantum chemistry leveraging sparsity and low rank factorization},
  author={Berry, Dominic W and Gidney, Craig and Motta, Mario and McClean, Jarrod R and Babbush, Ryan},
  journal={Quantum},
  volume={3},
  pages={208},
  year={2019},
  publisher={Verein zur F{\"o}rderung des Open Access Publizierens in den Quantenwissenschaften}
}

@article{babbush2018encoding,
  title={Encoding electronic spectra in quantum circuits with linear T complexity},
  author={Babbush, Ryan and Gidney, Craig and Berry, Dominic W and Wiebe, Nathan and McClean, Jarrod and Paler, Alexandru and Fowler, Austin and Neven, Hartmut},
  journal={Physical Review X},
  volume={8},
  number={4},
  pages={041015},
  year={2018},
  publisher={APS}
}

@article{zecchi2026block,
  title={Block encoding of sparse matrices with a periodic diagonal structure},
  author={Zecchi, Alessandro Andrea and Sanavio, Claudio and Cappelli, Luca and Perotto, Simona and Roggero, Alessandro and Succi, Sauro},
  journal={arXiv preprint arXiv:2602.10589},
  year={2026}
}

@article{kuklinski2025simpler,
  title={A simpler Gaussian state-preparation},
  author={Kuklinski, Parker and Rempfer, Benjamin and Obenland, Kevin and Elenewski, Justin},
  journal={arXiv preprint arXiv:2508.03987},
  year={2025}
}

@article{harrigan2024expressing,
  title={Expressing and analyzing quantum algorithms with qualtran},
  author={Harrigan, Matthew P and Khattar, Tanuj and Yuan, Charles and Peduri, Anurudh and Yosri, Noureldin and Malone, Fionn D and Babbush, Ryan and Rubin, Nicholas C},
  journal={arXiv preprint arXiv:2409.04643},
  year={2024}
}

@article{bungartz2004sparse,
  title={Sparse grids},
  author={Bungartz, Hans-Joachim and Griebel, Michael},
  journal={Acta numerica},
  volume={13},
  pages={147--269},
  year={2004},
  publisher={Cambridge University Press}
}

@article{sogabe2022krylov,
  title={Krylov subspace methods for linear systems},
  author={Sogabe, Tomohiro},
  journal={Springer Series in Computational Mathematics},
  year={2022},
  publisher={Springer}
}

@article{plesch2011quantum,
  title={Quantum-state preparation with universal gate decompositions},
  author={Plesch, Martin and Brukner, {\v{C}}aslav},
  journal={Physical Review A—Atomic, Molecular, and Optical Physics},
  volume={83},
  number={3},
  pages={032302},
  year={2011},
  publisher={APS}
}

@article{draper2004logarithmic,
  title={A logarithmic-depth quantum carry-lookahead adder},
  author={Draper, Thomas G and Kutin, Samuel A and Rains, Eric M and Svore, Krysta M},
  journal={arXiv preprint quant-ph/0406142},
  year={2004}
}

@article{bocharov2015efficient,
  title={Efficient synthesis of universal repeat-until-success quantum circuits},
  author={Bocharov, Alex and Roetteler, Martin and Svore, Krysta M},
  journal={Physical review letters},
  volume={114},
  number={8},
  pages={080502},
  year={2015},
  publisher={APS}
}

\end{document}